\documentclass{article}
\usepackage[utf8]{inputenc}
\usepackage{amsthm}
\usepackage{graphicx,ae,aecompl}
\usepackage{amssymb}
\usepackage{amsmath}
\usepackage{subfigure}
\usepackage{undertilde}
\usepackage{algorithmic}
\usepackage[boxruled,linesnumbered]{algorithm2e}
\usepackage[left=2.5cm,top=2.5cm,right=2.5cm,bottom=2.5cm]{geometry}
\usepackage{makeidx}
\usepackage{float}
\usepackage{scalefnt}
\newtheorem{prop}{Proposition}[section]
\newtheorem{lema}{Lemma}[section]

\title{Bayesian Wavelet Shrinkage with Beta Priors}
\author{Alex Rodrigo dos S. Sousa \\ USP, Brazil  \\ \and Nancy L. Garcia \\ Unicamp, Brazil \\ \and Brani Vidakovic \\ TAMU, USA }

\begin{document}

\numberwithin{equation}{section}
\numberwithin{table}{section}
\numberwithin{figure}{section}

 \maketitle
    \begin{abstract}
In wavelet shrinkage and thresholding, most of the standard techniques do not consider information that wavelet coefficients might be bounded, although information about bounded energy in signals can be readily available. To address this, we present a Bayesian approach for shrinkage of bounded wavelet coefficients in the context of non-parametric regression.  We propose the use of a zero-contaminated beta distribution with a support symmetric around zero as the prior distribution for the location parameter in the wavelet domain in models with additive gaussian errors. The hyperparameters of the proposed model are closely related to the shrinkage level, which facilitates their elicitation and interpretation. For signals with a low signal-to-noise ratio, the associated Bayesian shrinkage rules provide significant improvement in performance in simulation studies when compared with standard techniques.

Statistical properties such as bias, variance, classical and Bayesian risks of the associated shrinkage rules are presented and their performance is assessed in simulations studies involving standard test functions. Application to real neurological data set on spike sorting is also presented.
    \end{abstract}

\section{Introduction}

Wavelet-based methods are increasingly applied in many fields, such as mathematics, signal and image processing, geophysics, bioinformatics, and many others. In statistics, applications of wavelets arise mainly in the tasks involving non-parametric regression, density estimation, assessment of scaling, functional data analysis and stochastic processes. These methods basically utilize the possibility of representing functions that belong to certain functional spaces as expansions in a wavelet basis, similar to others expansions such as splines or Fourier, among others. However, the wavelet expansions have characteristics that make them particularly useful: they are localized in both time and scale/frequency in an adaptive way, their coefficients are typically sparse, the coefficients can be obtained by fast
computational algorithms, and the magnitudes of the coefficients can be linked to the smoothness properties of the functions they represent. These properties of wavelet representations enable adaptive time/frequency data analysis, bring computational advantages, and allow for statistical data modeling at varying resolution scales.

Wavelet shrinkage methods are used to estimate the underlying signal when its noisy version is observed. The noisy signal is transformed to a wavelet domain, the resulting wavelet coefficients are shrunk, and the inverse transform of the shrunk coefficients is taken as an estimator of the original signal. The wavelet shrinkage is already a mature research field, many techniques available in the literature. The main works in this area are of Donoho (1993a, 1993b), Donoho and Johnstone (1994a, 1994b, 1995), but also Donoho et al. (1995, 1996), Johnstone and Silverman (1997), Vidakovic (1998, 1999b) and Antoniadis et al. (2002) can be cited. For more details on shrinkage methods, see Vidakovic (1999a) and Jansen (2001). The statistical models in which shrinkage techniques are applied standardly assume Gaussian additive errors. These models are important not only because of their applicability to a range of different problems, but also from the mathematical point of view since the Gaussian additive errors remain both Gaussian and additive after the data are wavelet-transformed.

Bayesian shrinkage methods have also been extensively studied, mainly for the possibility of adding, by means of a prior probabilistic distributions,  available information about the regression, coefficients and model parameters to be estimated. Specifically in the case of wavelets, information about the degree of sparsity of the coefficients, the support of these coefficients, the signal smoothness, its self-similarity, and monotonicity, to list a few, can be incorporated into the statistical model by a Bayesian approach. To achieve this, the choice of the prior distribution in the model describing wavelets coefficients is critically important to achieve meaningful results.

Many Bayesian shrinkage procedures have been studied and proposed recently in many statistical fields. Some examples are Lian (2011), Berger et al. (2012), Karagiannis et al. (2015), Griffin and Brown (2017), and Torkamani and Sadeghzadeh (2017). Bayesian models in the wavelet domain were proposed since the late 1990s, see Chipman et all (1997), Abramovich et al. (1998), Vidakovic (1998), Vidakovic and Ruggeri (2001), Angelini and Vidakovic (2004), Johnstone and Silverman (2005), Reményi and Vidakovic (2015), Bhattacharya et al. (2015) among others.
Bayesian models in the wavelet domain have showed to be capable of
incorporating prior information about the unknown regression function such as
smoothness, periodicity, sparseness, self-similarity and, for some
particular bases (e.g., Haar),  monotonicity.

Although classical and Bayesian shrinkage/thresholding procedures available in the literature are well suited for many applications and real-data denoising tasks, typically, they do not incorporate information on the bounded energy in signals,
which is manifested as boundedness of wavelet coefficients in the wavelet domain.  Angelini and Vidakovic (2004) proposed only uniform and Bickel distributions as bounded priors in models on wavelet coefficients to incorporate information about bounded energy in signals. In this paper we propose a more general family of bounded priors useful to model bounded energy signals. As we will show later, taking boundedness of the coefficients into account, which is readily modeled by a Bayesian approach, would improve estimation of signals especially when the signal-to-noise ratio is low. Since wavelet coefficients are well localized, this improved estimation can better recover specific features of the unknown function in nonparametric regression model, such as peaks, cusps, or discontinuities.

 The prior information on the energy bound
often exists in real life problems. In the wavelet domain this can be modeled by the assumption that the location parameter in a model for a wavelet coefficient is bounded. Estimation of a bounded normal mean has
been considered in Miyasawa (1953),  Bickel (1981), Casella and Strawderman (1981),
and Vidakovic and DasGupta (1996).  In
our context, if the structure of the prior can be
supported by the analysis of the empirical distribution of the
wavelet coefficients. When prior knowledge about the signal-to-noise ratio (SNR)  is available,
then any symmetric and unimodal distribution supported on a
bounded set, say $[-m, m]$, could be a possible candidate for
the prior.

To address denoising of bounded energy signals, in this paper we propose and explore the beta family of distributions symmetric around zero as priors for the location parameter in a Gaussian model on wavelet coefficients. As traditionally done in this kind of analysis, the prior is contaminated by a point mass at 0. This added point mass at zero to a spread part of the prior facilitates the shrinkage and makes it adaptive.

Several reasons motivate the use of beta family. The flexibility of the beta distribution, as a spread part of prior, is readily controlled by convenient choice of its parameters. Moreover, we show that there is an interesting relationship between the (hyper) parameters and the degree of wavelet shrinkage, which is critical for denoising tasks. If the problem is rescaled so that the size of the noise (its variance) is 1, then $m$ can be taken as SNR. In this context, beta distribution is a proper choice due its boundedness and  flexibility, which is a provides advantage compared to the already proposed uniform and Bickel priors. 
Furthermore, the performance of the shrinkage rules under beta family was found to be superior to some of the traditionally used classical and Bayesian shrinkage/thresholding methods applied in practice. The considered scenarios in the simulation studies involve test functions with different features to be recovered, such as spikes,   discontinuities, and oscillations, which guarantees some independence of the characteristics of the signal to be estimated.
Finally, the performance of the proposed beta shrinkage for low SNR in simulated datasets and in real weak-signal datasets shows an advantage with respect to mean square error, when compared to some commonly employed methods.

In summary, the novelty of this paper in terms of methodology is the use and assessment of rescaled and zero-mass contaminated distributions from beta family as priors on wavelet coefficients, allowing for boundedness information in shrinkage procedure. Moreover, the beta distribution generalizes the already proposed uniform prior model and is related to Bickel prior (but with hyperparameters much easier to elicit and interpret). As an extension, we also propose a triangular prior model, as a convolution of two uniform priors and demonstrate its good performance in wavelet shrinkage tasks. 

The proposed wavelet shrinkage as a computationally straightforward task, is serving as a building block of more advanced computational methods, involving adaptive smoothing of noisy phenomena and dealing with potentially large data sets.  In this sense, the contribution of this paper  can be viewed as a potential building brick for tasks in computational statistical learning.

This paper is organized as follows: Section 2 defines the considered model in time and wavelet domain and the proposed prior distribution, Section 3 presents the shrinkage rule and its statistical properties such as variance, bias and risks. As an extension of the beta prior, we consider the shrinkage rule under triangular prior in Section 4. Section 5 is dedicated to prior elicitation. To demonstrate the performance of the proposed approach simulation studies are presented in Section 6,
 and the shrinkage rule is applied in a spike-sorting real data set in Section 7. Section 8 provides conclusions.

\subsection{A Motivational Example}
As a motivational example we show the application and advantages of the proposed method. As signal, we consider Donoho and Johnstone's test function Doppler defined by $f(x) = \sqrt{x(1-x)}\sin\left(\frac{2.1\pi}{x+0.05}\right)$, $0 \leq x \leq 1$. Note that Doppler function is a bounded energy signal. Furthermore, it has oscillations as an interesting feature to be preserved in the task of denoising.

We evaluated the Doppler function at 512 equispaced points and added a random normal noise with signal-to-noise ratio SNR = 3, as shown in Figure \ref{fig:killer}(a). The goal is to recover the true signal by applying a wavelet shrinkage.

Figure \ref{fig:killer}(b) shows the signal (Doppler function) and the recovered signal (estimated function) after applying the shrinkage rule on empirical wavelet coefficients under the proposed mixture of a point mass function at zero and beta distribution as a prior. In fact, it will be shown in Section 6 (Simulation Studies) that this shrinkage rule has the best performance in terms of averaged mean squared error when compared to standard techniques. Moreover, this shrinker incorporates the prior knowledge about boundedness of the signal, which is not the case for the standard techniques. In addition, the shrinkage rule is simple.

\begin{figure}[H]
\centering
\subfigure[Generated data.]{
\includegraphics[scale=0.45]{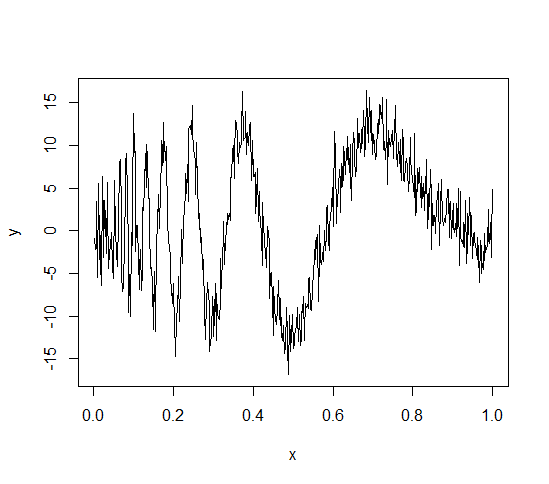}}
\subfigure[True and estimated signals.]{
\includegraphics[scale=0.45]{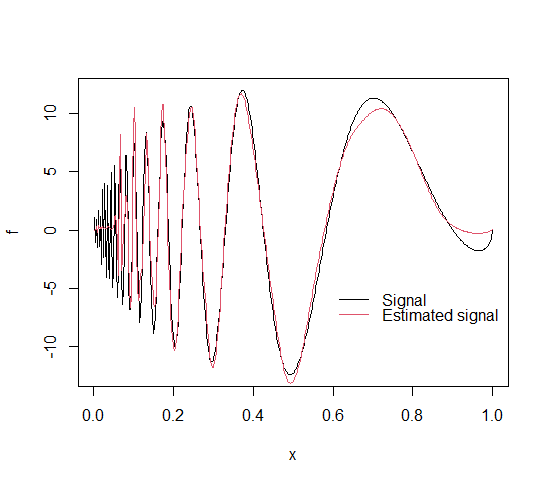}}
\caption{(a) Noisy Doppler of sizr 512 and low signal-to-noise ratio (SNR=3). (b) Estimated Doppler signal by shrinkage rule defined by the mixture of a point mass at zero and rescaled beta distribution as a prior. The origiunal sigmal is also shown for reference.} \label{fig:killer}
\end{figure}

\section{The Model}
\subsection{The Symmetric around Zero Beta Distribution}
In statistics, the beta family of distributions is extensively used to model random phenomena confined to the $[0,1]$ domain. The shape of beta distribution is very flexible; it is controlled by convenient choice of its parameters. In our framework, it is necessary to use its transformed version.
Because of symmetry of the support, the distribution is shifted and rescaled to the interval $[-m,m]$.  We also want to choose its parameters to keep it symmetric about 0. This requires both parameters of beta to have the same value. Therefore, we propose the use of beta distribution with support symmetric around zero as the spread-part for the prior distribution in the Bayesian model on the location of wavelet coefficients. Its density function is
\begin{equation}\label{eq:beta}
g(x;a,m) = \frac{(m^2 - x^2)^{(a-1)}}{(2m)^{(2a-1)}B(a,a)}\mathbb{I}_{[-m,m]}(x),
\end{equation}

\noindent where $B(\cdot , \cdot)$ is the standard beta function, $a>0$ and $m>0$ are parameters, and $\mathbb{I}_{[-m,m]}(\cdot)$ is an indicator function equal to 1 when its argument falls in the interval $[-m,m]$ and 0 otherwise.

 For $a>1$, the density function (\ref{eq:beta}) is unimodal around zero and as $a$ increases, the density becomes more concentrated around zero. This is an important feature for wavelet shrinkage methods, since high values of $a$ imply higher levels of shrinkage, which results in sparse estimated coefficients. Density \eqref{eq:beta} becomes uniform for $a=1$, which was already considered by Angelini and Vidakovic (2004). In this work we consider beta densities with $a \geq 1$.  Figure \ref{fig:beta} shows the translated and rescaled beta density for some selected values of parameter $a$ in the interval $[1,10]$ and $m=3$.

\begin{figure}[H]
\centering
\includegraphics[scale=0.50]{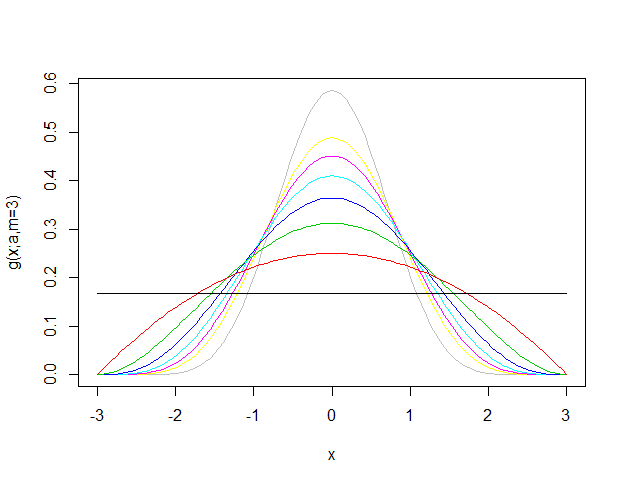}
\caption{Shifted and rescaled beta densities for some values of $a \in [1,10]$ and $m=3$.}\label{fig:beta}
\end{figure}
\subsection{Zero-contaminated Beta Distribution as a Prior}

We start with the nonparametric regression problem of the form
\begin{equation} \label{eq:modeltime}
y_i = f(x_i) + e_i , \qquad i=1,...,n=2^J, J \in \mathbb{N},
\end{equation}
\noindent where $f \in \mathbb{L}_2(\mathbb{R})= \{f:\int f^2 < \infty\}$ and $e_i$, $i=1,...,n$, are zero mean independent normal random variables with unknown variance $\sigma^2$. In vector notation, we have
\begin{equation}\label{eq:modeltimevec}
\boldsymbol{y} = \boldsymbol{f} + \boldsymbol{e},
\end{equation}

\noindent where $\boldsymbol{y} = (y_1,...,y_n)'$, $\boldsymbol{f} = (f(x_1),...,f(x_n))'$ and $\boldsymbol{e} = (e_1,...,e_n)'$. The goal is to estimate the unknown function $f$. After applying a discrete wavelet transform (DWT) on (2.3), given by the corresponding orthogonal matrix $D$, in the wavelet domain we obtain the following model,
\begin{equation} \label{eq:modelvec}
\boldsymbol{d} = \boldsymbol{\theta} + \boldsymbol{\epsilon},
\end{equation}
where
$\boldsymbol{d} = D\boldsymbol{y}$, $\boldsymbol{\theta} = D \boldsymbol{f}$ and $\boldsymbol{\epsilon} = D \boldsymbol{e}$.

Due to the independence of the random errors and the orthogonality of the $D$ transform, the model in the wavelet domain does not change its statistical structure. It remains additive and the errors are i.i.d. normal.

Because  the strong decorrelating property of wavelet transforms we can model one coefficient at a time. For the $i$th component of the vector $\boldsymbol{d}$, we have a simple model
\begin{equation}\label{eq:model}
d_i = \theta_i + \epsilon_i,
\end{equation}
\noindent where $d_i$ is the empirical wavelet coefficient, $\theta_i \in [-m,m]$ is the coefficient to be estimated and $\epsilon_i \sim N(0,\sigma^2)$ is the normal random error with unknown variance $\sigma^2$. For the simplicity of notation, we suppress the indices in $d$, $\theta$ and $\epsilon$. Note that, according to the model \eqref{eq:model}, $d|\theta \sim N(\theta,\sigma^2)$ and the problem of estimating a function $f$ becomes a normal mean estimation problem in the wavelet domain for each coefficient. Once this bounded mean estimation problem is solved, the vector $\boldsymbol{f}$ can be estimated by the application of the inverse wavelet transform on $\boldsymbol{\hat{\theta}}$.

To complete the Bayesian model, we propose the following prior distribution for $\theta$,
\begin{equation} \label{eq:prior}
\pi(\theta;\alpha,a,m) = \alpha \delta_{0}(\theta) + (1-\alpha)g(\theta;a,m),
\end{equation}
where $\alpha \in (0,1)$, $\delta_{0}(\theta)$ is the point mass function at zero and $g(\theta;a,m)$ is the beta distribution \eqref{eq:beta} in $[-m,m]$. The proposed prior distribution has $\alpha \in (0,1)$, $a>0$ and $m>0$ as hyperparameters and their choices are directly related to the degree of shrinkage of the empirical coefficients. It will be shown that as $a$ or $\alpha$ (or both of them) increase, the level of shrinkage increases as well.

\section{The Shrinkage Rule and its Statistical Properties}

The shrinkage rule $\delta(\cdot)$ for Bayesian estimation of the wavelet coefficient $\theta$ of model \eqref{eq:model} depends on the choice of location of the posterior (mean, mode, or median) and the loss function. Under square error loss function $L(\delta,\theta) = (\delta - \theta)^2$, it is well known that the Bayes rule is the posterior expected value of $\theta$, i.e, $\delta(d) = \mathbb{E}_{\pi}(\theta \mid d)$ minimizes the Bayes risk. The Proposition \ref{prop1} gives an expression of the shrinkage rule under a mixture  prior consisting of  a point mass at zero and a density function with support in $[-m,m]$.

\begin{prop} \label{prop1}
If the prior distribution of $\theta$ is of the form $\pi(\theta;\alpha,m) = \alpha \delta_{0}(\theta) + (1-\alpha)g(\theta)$, where $g$ is a density function with support in $[-m,m]$, then the shrinkage rule under the quadratic loss function is given by
\begin{equation}
\delta(d) = \frac{(1-\alpha)\int_{\frac{-m-d}{\sigma}}^{\frac{m-d}{\sigma}}(\sigma u + d)g(\sigma u + d)\phi(u)du}{\alpha \frac{1}{\sigma}\phi(\frac{d}{\sigma})+(1-\alpha)\int_{\frac{-m-d}{\sigma}}^{\frac{m-d}{\sigma}}g(\sigma u + d)\phi(u)du}
\end{equation}
where $\phi(\cdot)$ is the standard normal density function.
\end{prop}

\begin{proof}
If $\mathcal{L}(\cdot \mid \theta)$ is the likelihood function, we have that
\begin{align*}
\delta(d) &= \mathbb{E}_{\pi}(\theta \mid d) \\
          &=\frac{\int_{\Theta}\theta[\alpha\delta_{0}(\theta)+(1-\alpha)g(\theta)]\mathcal{L}(d \mid \theta)d\theta}{\int_{\Theta}[\alpha\delta_{0}(\theta)+(1-\alpha)g(\theta)]\mathcal{L}(d \mid \theta)d\theta} \\
          &= \frac{(1-\alpha)\int_{-m}^{m}\theta g(\theta)\frac{1}{\sqrt{2\pi}}\exp\{-\frac{1}{2}(\frac{d-\theta}{\sigma})^2\}\frac{d\theta}{\sigma}}{\alpha \frac{1}{\sigma\sqrt{2\pi}}\exp\{-\frac{1}{2}(\frac{d}{\sigma})^2\}+(1-\alpha)\int_{-m}^{m}g(\theta)\frac{1}{\sqrt{2\pi}}\exp\{-\frac{1}{2}(\frac{d-\theta}{\sigma})^2\}\frac{d\theta}{\sigma}}\\
          &= \frac{(1-\alpha)\int_{\frac{-m-d}{\sigma}}^{\frac{m-d}{\sigma}}(\sigma u + d)g(\sigma u + d)\phi(u)du}{\alpha \frac{1}{\sigma}\phi(\frac{d}{\sigma})+(1-\alpha)\int_{\frac{-m-d}{\sigma}}^{\frac{m-d}{\sigma}}g(\sigma u + d)\phi(u)du}.\\
\end{align*}
\end{proof}

Although Proposition 3.1 gives us a general expression for shrinkage rule under a prior model of the form
\eqref{eq:prior}, an explicit formula for the shrinker when beta distribution is considered as $g$ is not available, in general. For this reason, we obtain the shrinkage rules and its properties, such bias and risks numerically, using computational methods based on standard Monte Carlo techniques for calculating the involved integrals in the expression of Proposition 3.1. The R codes of the shrinkage rules and their properties are available in our R package \textit{bayesShrink}, which is under development but already available at \textit{Github} repository, see Sousa et al. (2020).

Figure \ref{fig:shrink} presents some shrinkage rules for $g$ as beta distribution \eqref{eq:beta} with hyperparameters $m = 3$, $\alpha = 0.9$ and for some values of $a \in [1,10]$ as well as their variances. It can be seen that the length of the interval in which the rule shrinks to zero increases as the hyperparameter $a$ increases, since high values of $a$ result in more concentrated beta distributions around zero. A typical feature of these rules is that as $d$ increases, $\delta(d)$ gets closer to $m$ and as $d$ decreases, $\delta(d)$ gets closer to $-m$. These asymptotic characteristics are reasonable since there is the assumption that the coefficients to be estimated belong to the range $[-m, m]$, so the empirical coefficients outside this range may occur only due to the presence of noise.

\begin{figure}[H]
\centering
\subfigure[Shrinkage rules\label{lognormal}]{
\includegraphics[scale=0.45]{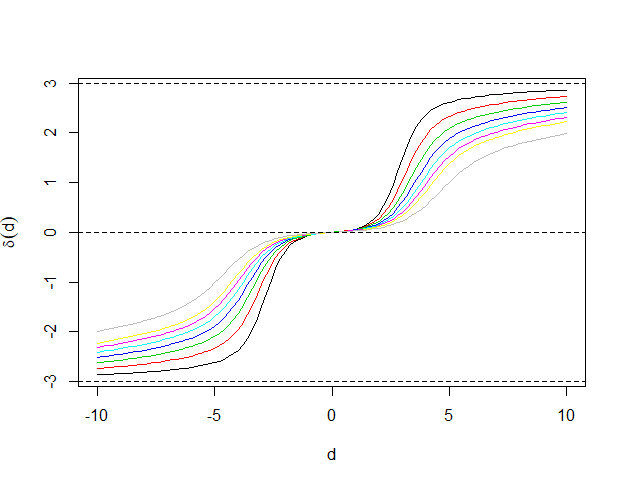}}
\subfigure[Variances\label{blocls}]{
\includegraphics[scale=0.45]{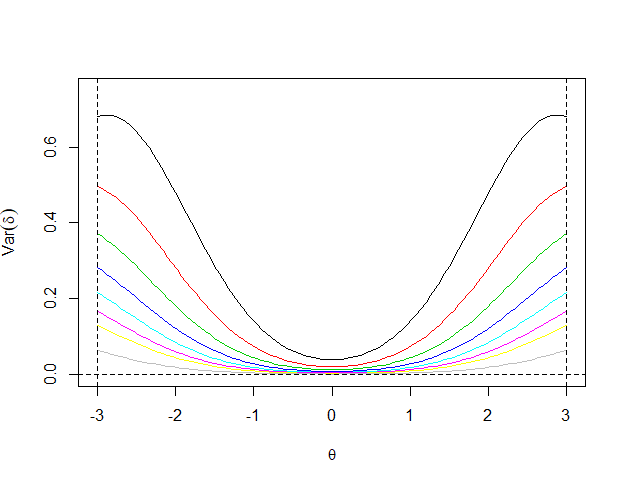}}
\caption{Shrinkage rules and their variances under beta prior distribution with hyperparameters $m=3$, $\alpha = 0.9$ and values of $a \in [1,10]$.} \label{fig:shrink}
\end{figure}

Figure \ref{fig:bias} (a) and (b) shows the squared bias and classical risks (denoted by $R(\theta)$) respectively for the same shrinkage rules considered above. Observe that, as expected, the rules have smaller variances and biases for values of $\theta$ near zero, reaching minimum values in both graphs when the wavelet coefficient is zero. It is also noted that as hyperparameter $a$ increases, the bias of the estimator increases and the variance decreases. The classical risk decreases as $\theta$ tends to zero and that for high values of $\theta$, the risk is larger for rules with large values of $a$. These features are justified by the fact that the degree of shrinkage increases as the hyperparameter $a$ increases, so if the value of the wavelet coefficient is far from zero, such rules with larger values of $a$ tend to underestimate $\theta$ than rules with small values of $a$.

\begin{figure}[H]
\centering
\subfigure[Squared bias\label{lognormal}]{
\includegraphics[scale=0.4]{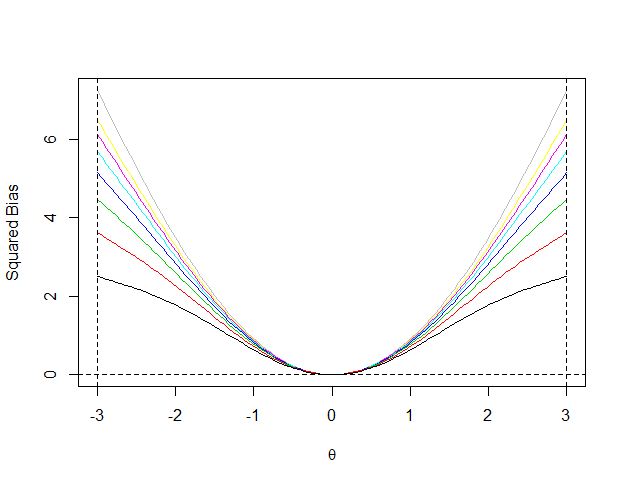}}
\subfigure[Classical Risks\label{blocls}]{
\includegraphics[scale=0.4]{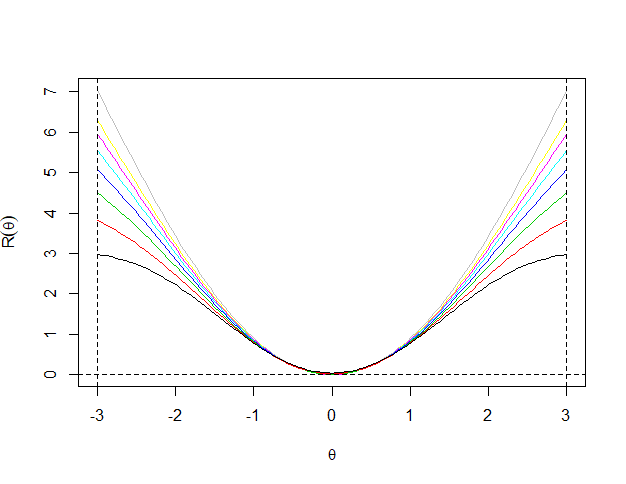}}
\caption{Squared bias and classical risks of the shrinkage rules under beta prior distribution with hyperparameters $m=3$, $\alpha = 0.9$ and values of $a \in [1,10]$.} \label{fig:bias}
\end{figure}

Finally, Tables \ref{tab:brisk} and \ref{tab:brisk2} show Bayes risks (denoted $r_{\delta}$) in terms of the hyperparameters $a$ and $\alpha$ respectively of the shrinkage rules considered. As expected, Bayes risk decreases as either $a$ or $\alpha$ increases.

\begin{table}[!htb]
\centering
\label{my-label}
\begin{tabular}{ccccccccc}
\hline
$a$ & 1   & 2     & 3  & 4 & 5 & 6 & 7 & 10     \\ \hline
$r_{\delta}$ & 0.189 & 0.137 & 0.101 & 0.088 & 0.074 & 0.063 & 0.056 & 0.041 \\ \hline
\end{tabular}
\caption{Bayes risks of the shrinkage rules under beta prior distribution with hyperparameters $m=3$ and $\alpha = 0.9$.}\label{tab:brisk}
\end{table}

\begin{table}[!htb]
\centering
\label{my-label}
\begin{tabular}{ccccccc}
\hline
$\alpha$ & 0.6   & 0.7 & 0.8 & 0.9 & 0.99     \\ \hline
$r_{\delta}$ & 0.399 & 0.326 & 0.241 & 0.137 & 0.017 \\ \hline
\end{tabular}
\caption{Bayes risks of the shrinkage rules under beta prior distribution with hyperparameters $m=3$ and $a=2$.}\label{tab:brisk2}
\end{table}
\section{An Extension: The Triangular Prior}

We briefly present the triangular prior distribution for the wavelet coefficients as an extension, since its associated shrinkage rule has explicit formula in terms of the standard normal density and cumulative functions. In fact, the triangular distribution, popularly referred as ``witch hat'', on $[- m, m]$ is the convolution of two uniform distributions on $[-m/2, m/2]$ and its density is given by
\begin{equation} \label{eq:triang}
g_{T}(x;m)=\left\{\begin{array}{rc}
\frac{x+m}{m^2},&\mbox{if}\quad -m \leq x < 0,\\
\frac{m-x}{m^2}, &\mbox{if}\quad 0 \leq x \leq m, \\
0, &\mbox{}\quad else.
\end{array}\right.
\end{equation}

The following proposition provides an explicit formula for the shrinkage rule under the triangular prior.

\begin{prop}\label{prop4}
The shrinkage rule under prior distribution of the form $\pi(\theta;\alpha,m) = \alpha \delta_{0}(\theta) + (1-\alpha)g_{T}(\theta;m)$, where $g_{T}(\cdot;m)$ is the triangular distribution over $[-m,m]$, is

\begin{equation}
\delta_T(d) = \frac{(1-\alpha)S_1(d)}{\frac{\alpha m^2}{\sigma}\phi(\frac{d}{\sigma}) + (1-\alpha)S_2(d)},
\end{equation}
where

\noindent $S_1(d) = d\sigma[\phi(\frac{m+d}{\sigma}) + \phi(\frac{m-d}{\sigma}) -2\phi(\frac{d}{\sigma})] + (d^2+\sigma^2+dm)\Phi(\frac{m+d}{\sigma}) +(d^2+\sigma^2- dm)\Phi(\frac{m-d}{\sigma})-2(d^2+\sigma^2)\Phi(\frac{d}{\sigma})$,

\noindent $S_2(d) = \sigma[\phi(\frac{m+d}{\sigma}) + \phi(\frac{m-d}{\sigma}) -2\phi(\frac{d}{\sigma})]+(d+m)\Phi(\frac{m+d}{\sigma})+(d-m)\Phi(\frac{m-d}{\sigma})-2d\Phi(\frac{d}{\sigma})$,

\noindent for $\phi(\cdot)$ and $\Phi(\cdot)$ the standard normal density and cumulative distribution respectively.

\end{prop}

\begin{proof}
Applying Proposition \ref{prop1} for the density \eqref{eq:triang} and solving the integrals directly. We provide the details in Appendix A.
\end{proof}

Shrinkage rules under triangular prior with $m=3$ and $\alpha \in \{0.6, 0.7, 0.8, 0.9, 0.99\}$ and their statistical properties are shown in Figure \ref{fig:triang} and Table \ref{tab:trisk}. The behaviors of the rules and the properties are the same as the rules under beta prior. The performances of the shrinkage rule under triangular prior in our simulation studies were great as we will see later and its explicit formula can bring advantages in computational implementation.

\begin{figure}[H]
\centering
\includegraphics[scale=0.70]{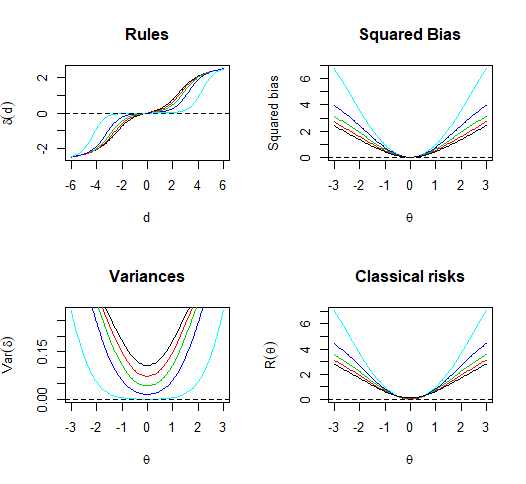}
\caption{Shrinkage rules (top left), squared bias (top right), variances (bottom left) and classical risks (bottom right) for triangular prior with $m=3$ and $\alpha \in \{0.6, 0.7, 0.8, 0.9, 0.99\}$.}\label{fig:triang}
\end{figure}

\begin{table}[!htb]
\centering
\label{my-label}
\begin{tabular}{ccccccc}
\hline
$\alpha$ & 0.6   & 0.7 & 0.8 & 0.9 & 0.99     \\ \hline
$r_{\delta}$ & 0.357 & 0.289 & 0.212 & 0.119 & 0.014 \\ \hline
\end{tabular}
\caption{Bayes risks of the shrinkage rules under triangular prior distribution with hyperparameters $m=3$ and $\alpha \in \{0.6, 0.7, 0.8, 0.9, 0.99\}$.}\label{tab:trisk}
\end{table}

\section{Default Prior Hyperparameters}

Methods and criteria for determination of the associated model hyperparameters to estimate the coefficients are critical for any Bayesian procedure. In the framework of Bayesian shrinkage with a beta prior, the choices of the $\sigma$ parameter of the random error distribution and the hyperparameters $\alpha$, $m$ and $a$ of the beta prior distribution are required. We present the methods and criteria adapted from proposals already available in the literature. The proposed metods for parameter or hyperparameter selection are used in simulation and application studies.

Based on the fact that much of the noise information present in the data can be obtained at the fine-resolutions in a wavelet decomposition, for the robust $\sigma$ estimation, Donoho and Johnstone (1994a) suggest

\begin{equation}\label{eq:sigma}
\hat{\sigma} = \frac{\mbox{median}\{|d_{J-1,k}|:k=0,...,2^{J-1}\}}{0.6745},
\end{equation}
where $J$ is the finest multiresolutuion level in the wavelet decomposition.
This choice is popular in all ``plug-in'' methods where th parameter $\sigma$ is estimated directly from the data.

The  hyperparameters $\alpha$ and $m$ are suggested as dependent on the level of resolution $j$ according to the expressions
\begin{equation}\label{eq:alpha}
\alpha = \alpha(j) = 1 - \frac{1}{(j-J_{0}+1)^\gamma}
\end{equation}
and
\begin{equation}\label{eq:m}
m = m(j) = \max_{k}\{|d_{jk}|\},
\end{equation}
where $J_ 0 \leq j \leq J-1$, $J_0$ is the coarsest resolution level and $\gamma > 0$. This was suggested by Angelini and Vidakovic (2004), who also indicate that in the absence of additional information, $\gamma = 2$ can be adopted as universal choice.

Many methods for choosing the hyperparameters in beta priors are available in the literature. Chaloner et al. (1983) proposed a method for choosing the hyperparameters based on the probability of success in Bernoulli trials. Duran and Booker (1988) used a percentile method. For $k \in [-m, m]$ and $p \in (0,1) $ fixed, $a$ is chosen so that
$P(\theta \leq k)=p$, i.e,
\begin{equation}\label{eq:a}
\int_{-m}^{k}\frac{(m^2 - \theta^2)^{(a-1)}}{(2m)^{(2a-1)}B(a,a)}d\theta = p. \nonumber
\end{equation}

Thus, the choice of $a$ is made by determining the probability of occurrence of the event $\{\theta \leq k \}$. This procedure is interesting because it uses subjective determination of probability, that is, it is  cognitively simpler to assign a probability to a certain event than to directly assign a value to the hyperparameter. In this work, however, we choose $a$ according to the desired shrinkage level of the empirical coefficients. As discussed in the paper, the shrinkage level increases as $a$ increases, since corresponding beta prior becomes more concentrated around zero.

Another possibility which adds to adaptivity is to consider the hyperparameter $a$ as level dependent, i.e, $a = a(j)$. However, we achieve adaptivity by a level-dependednt parameter $\alpha(j)$ which is a weight of point mass at zero in a contamination prior, and we resort to a fixed choice of $a$ in simulations and applications.

\section{Simulation Studies}

Simulation studies based on intensive computing were done to evaluate the performance of the shrinkage rules under beta priors for the particular cases in which the hyperparameter $a$ assumes the fixed values $a=1, 2, 5, and 10$ and triangular distribution (Triang). We compared these choices with the performances of some of the popular shrinkage/thresholding methods in the literature, namely, universal thresholding (Univ) proposed by Donoho and Johnstone (1994), false discovery rate (FDR) proposed by Abramovich and Benjamini (1996), cross validation (CV) of Nason (1996), Stein unbiased risk estimator (SURE) of Donoho and Johnstone (1995), Bayesian adaptive multiresolution shrinker (BAMS) of Vidakovic and Ruggeri (2001) and larger posterior mode (LPM) of Cutillo et al. (2008). We also considered the shrinkage rule under Bickel prior, suggested by Angelini and Vidakovic (2004). They proved that the shrinkage rule under this prior is approximately $\Gamma$-minimax for the class of all symmetric unimodal priors bounded on $[-m, m]$, $\Gamma_{SU[-m,m]}$ .
 Bickel (1981) proved that, when $m$ increases, the weak limit of the least favorable prior in $\Gamma_{SU[-m,m]}$ is approximately (in sense of weak distributional limits)
 $g_m(\theta)=\frac{1}{m}\cos^2\left(\frac{\pi \theta}{2m}\right)
\mathbb{I}_{[-m,m]}(\theta)$. Applying this result in our context, we
have that the Bickel shrinkage is induced by prior
\begin{equation*} \label{mlarge} \pi(\theta)=\alpha
\delta_0+(1-\alpha)\frac{1}{m}\cos^2\left(\frac{\pi \theta}{2m}\right)
\mathbb{I}_{[-m,m]}(\theta).
\end{equation*}
The corresponding Bayes rule does not have a simple analytical form and needs to be numerically computed.
The hyperparameters $m$ and $\alpha$ were selected according to the proposals described in Section 5.

To perform the simulation, the rules were applied in the Donoho-Johnstone (DJ) test functions. These functions, shown in Figure \ref{fig:dj},  are commonly used in the literature for comparison of wavelet-based shrinkage methods. The four functions, called Bumps, Blocks, Doppler and Heavisine, have specific features that mimic features occuring in practice: oscilations in Doppler, spikes in Bumps, discontinuities in Blocks, and cusps in Heavisine.
\begin{figure}[H]
\centering
\includegraphics[scale=0.50]{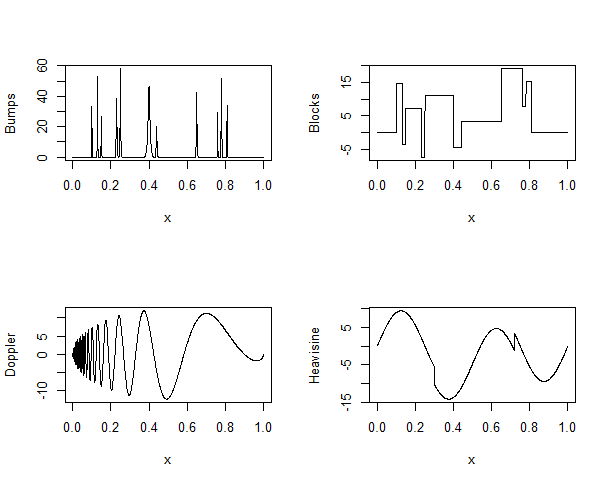}
\caption{Donoho-Johnstone (DJ) test functions.}\label{fig:dj}
\end{figure}

For each test function $f$, three sample sizes were selected, $n = 512, 1024$ and $2048$, and i.i.d. normals with zero mean and variance $\sigma^2$ were added, with $\sigma^2$ selected according to three signal to noise ratio (SNR), 3, 5 and 7. We used  Daubechies filter with ten vanishing moments (Daub10) to transform the noisy signals. After the shrinkage/thresholding procedure was applied, the processed coefficients are transformed back by inverse wavelet transform to the domain of the original signals.
This inverse transformed signal is the esstimator $\hat f$ of the original test function $f$.

We used the mean squared error (MSE), $MSE = \frac{1}{n} \sum_{i=1}^{n}[{\hat f(x_i)} - f(x_i)]^2$ as to compare the shrinkage rules. For each function (each $n$ and SNR), the process was repeated $M = 200$ times and the average of the resulting MSEs, $AMSE = \frac{1}{M} \sum_{j=1}^{M}MSE_j$, was calculated, as shown in Tables \ref{tab:sim1} and \ref{tab:sim2}. Figure \ref{fig:fits} presents the estimates produced by the shrinkage rule under beta prior with $a=2$ and $n=2048$.

In general, the beta prior and triangular shrinkage rules had superior performances in the simulations. They dominated Univ, FDR, CV, SURE, BAMS, LPM and Bickel prior shrinkage in practically all the scenarios. We highlight the good results of the beta prior rules for lower SNR (SNR=3) and sparsly sampled signal ($n=512$) when compared with the other methods, which is good motivation to use this shrinkage in denoising tasks with real data. In Tables \ref{tab:sim1} and \ref{tab:sim2} the best results are in bold. Another point to be emphasized is solid performance of the beta and triangular shrinkage rules even in the scenarios where they were not dominant. The overall simulation results indicate good flexibility of the proposed shrinkage rules in estimation of functions with different characteristics.

It is worth mentioning, that the AMSE decreased for $a=1,2,$ and $5$ and then increased for $a=10$ in most of the scenarios. As the shrinkage level increases when $a$ increases, for large values of $a$ the empirical coefficients are excessively shrunk. This results in oversmoothing of important features, such spikes, cusps or discontinuities, negatively affecting AMSE performance.

\begin{table}[H]
\scalefont{0.5}
\centering
\label{my-label}
\begin{tabular}{|c|c|c|c|c|c|||c|c|c|c|c|c|}
\hline

Signal & n & Method & SNR=3 & SNR=5 & SNR=7&Signal&n&Method&SNR=3&SNR=5&SNR=7  \\ \hline \hline

Bumps&512	&	Univ	&	11.080	&	5.170	&	3.026	& Blocks	&512	&	Univ	&	6.928	&	3.660	&	2.254	\\
&	&	FDR	&	9.291	&	4.373	&	2.630	&	&	&	FDR	&	5.896	&	2.903	&	1.746	\\
&	&	CV	&	11.389	&	9.406	&	6.292	&	&	&	CV	&	2.559	&	1.250	&	0.841	\\
&	&	SURE	&	3.609	&	1.556	&	0.882	&	&	&	SURE	&	2.766	&	1.216	&	0.679	\\
&	&	BAMS	&	2.867	&	1.528	&	1.265	&	&	&	BAMS	&	2.465	&	1.297	&	1.091	\\
&	&	LPM	&	4.892	&	1.960	&	1.000	&	&	&	LPM	&	4.892	&	1.960	&	1.000	\\
&	&	Bickel	&	2.814	&	1.156	&\textbf{0.654}	&	&	&	Bickel	&	2.748	&	1.191	&	1.590	\\
&	&	$a=1$	&	2.995	&	1.238	&	0.696	&	&	&	$a=1$	&	2.915	&	1.315	&	0.684	\\
&	&	$a=2$	&	2.874	&	1.182	&	0.674	&	&	&	$a=2$	&	2.799	&	1.311	&	0.799	\\
&	&	$a=5$	& \textbf{2.812}	&\textbf{1.144}	&	0.661	&	&	&	$a=5$	&\textbf{2.687}	&	\textbf{1.181}	&\textbf{0.617}	\\
&	&	$a=10$	&	2.936	&	1.186	&	0.720	&	&	&	$a=10$	&	2.972	&	1.445	&	0.878	\\
&	&	Triang	&	2.828	&	1.157	&	0.656	&	&	&	Triang	&	2.727	&	1.249	&	0.786	\\ \hline
																			
&1024	&	Univ	&	7.547	&	3.570	&	2.128	&	&1024	&	Univ	&	4.848	&	2.479	&	1.542	\\
&	&	FDR	&	5.556	&	2.524	&	1.473	&	&	&	FDR	&	3.896	&	1.874	&	1.125	\\
&	&	CV	&	2.924	&	1.925	&	1.719	&	&	&	CV	&	1.789	&	0.838	&	0.533	\\
&	&	SURE	&	2.467	&	1.057	&	0.590	&	&	&	SURE	&	1.888	&	0.837	&\textbf{0.474}	\\
&	&	BAMS	&	2.155	&	1.046	&	0.860	&	&	&	BAMS	&	1.856	&	0.842	&	0.686	\\
&	&	LPM	&	4.966	&	1.957	&	0.998	&	&	&	LPM	&	4.966	&	1.957	&	0.998	\\
&	&	Bickel	&	1.972	&	0.978	&\textbf{0.466}	&	&	&	Bickel	&	1.718	&	1.788	&	0.523	\\
&	&	$a=1$	&	2.099	&	0.958	&	0.506	&	&	&	$a=1$	&	1.852	&	0.821	&	0.562	\\
&	&	$a=2$	&	2.014	&	1.016	&	0.486	&	&	&	$a=2$	&	1.770	&	0.933	&	0.536	\\
&	&	$a=5$	&\textbf{1.949}	&\textbf{0.877}	&	0.546	&	&	&	$a=5$	&\textbf{1.713}	&\textbf{0.761}	&	0.514	\\
&	&	$a=10$	&	1.949	&	1.059	&	0.821	&	&	&	$a=10$	&	1.900	&	0.786	&	0.496	\\
&	&	Triang	&	1.963	&	0.902	&	0.476	&	&	&	Triang	&	1.728	&	0.907	&	0.508	\\ \hline
																			
&2048	&	Univ	&	5.042	&	2.343	&	1.389	&	&2048	&	Univ	&	3.417	&	1.772	&	1.101	\\
&	&	FDR	&	3.567	&	1.581	&	0.915	&	&	&	FDR	&	2.676	&	1.288	&	0.764	\\
&	&	CV	&	1.602	&	0.734	&	0.477	&	&	&	CV	&	1.301	&	0.588	&	0.353	\\
&	&	SURE	&	1.647	&	0.696	&	0.389	&	&	&	SURE	&	1.356	&	0.596	&\textbf{0.341}	\\
&	&	BAMS	&	1.635	&	0.667	&	0.527	&	&	&	BAMS	&	1.502	&	0.585	&0.459	\\
&	&	LPM	&	4.955	&	1.957	&	0.998	&	&	&	LPM	&	4.955	&	1.957	&	0.998	\\
&	&	Bickel	&	1.208	&\textbf{0.510}	&\textbf{0.320}	&	&	&	Bickel	&	1.566	&	0.602	&	1.393	\\
&	&	$a=1$	&	1.267	&	0.571	&	0.333	&	&	&	$a=1$	&	1.252	&	0.651	&	1.857	\\
&	&	$a=2$	&	1.227	&	0.530	&	0.337	&	&	&	$a=2$	&	1.430	&	0.608	&	1.592	\\
&	&	$a=5$	&	1.209	&	0.742	&	1.113	&	&	&	$a=5$	&\textbf{1.163}	&\textbf{0.573}	&	1.306	\\
&	&	$a=10$	&	1.311	&	0.580	&	0.339	&	&	&	$a=10$	&	1.480	&	0.723	&	1.366	\\
&	&	Triang	&\textbf{1.205}	&	0.518	&	0.354	&	&	&	Triang	&	1.381	&	0.582	&	1.457	\\ \hline

\end{tabular}
\caption{AMSE of the shrinkage/thresholding rules in the simulation study for Bumps and Blocks DJ test functions.}\label{tab:sim1}
\end{table}

\begin{table}[H]
\scalefont{0.5}
\centering
\label{my-label}
\begin{tabular}{|c|c|c|c|c|c|||c|c|c|c|c|c|}
\hline

Signal & n & Method & SNR=3 & SNR=5 & SNR=7&Signal&n&Method&SNR=3&SNR=5&SNR=7  \\ \hline \hline
Doppler&512	&	Univ	&	2.680	&	1.413	&	0.892	&Heavisine	&512	&	Univ	&	0.567	&	0.404	&	0.304	\\
&	&	FDR	&	2.565	&	1.259	&	0.767	&	&	&	FDR	&	0.595	&	0.436	&	0.312	\\
&	&	CV	&	1.293	&	0.647	&	0.451	&	&	&	CV	&\textbf{0.505}	&\textbf{0.279}	&\textbf{0.178}	\\
&	&	SURE	&	1.329	&	0.596	&	0.337	&	&	&	SURE	&	0.571	&	0.414	&	0.317	\\
&	&	BAMS	&	1.551	&	0.628	&	0.503	&	&	&	BAMS	&	1.153	&	0.327	&	0.233	\\
&	&	LPM	&	4.892	&	1.960	&	1.000	&	&	&	LPM	&	4.892	&	1.960	&	1.000	\\
&	&	Bickel	&	1.112	&\textbf{0.520}	&\textbf{0.297}	&	&	&	Bickel	&	0.896	&	0.631	&	0.464	\\
&	&	$a=1$	&	1.138	&	0.567	&	0.305	&	&	&	$a=1$	&	0.788	&	0.548	&	0.470	\\
&	&	$a=2$	&	1.117	&	0.542	&	0.303	&	&	&	$a=2$	&	0.832	&	0.578	&	0.447	\\
&	&	$a=5$	&	1.130	&	0.522	&	0.298	&	&	&	$a=5$	&	0.976	&	0.648	&	0.721	\\
&	&	$a=10$	&	1.274	&	1.551	&	2.113	&	&	&	$a=10$	&	1.230	&	0.718	&	0.495	\\
&	&	Triang	&\textbf{1.104}	&	0.525	&	0.311	&	&	&	Triang	&	0.837	&	0.561	&	0.444	\\ \hline
																			
&1024	&	Univ	&	1.612	&	0.846	&	0.534	&	&1024	&	Univ	&	0.460	&	0.314	&	0.231	\\
&	&	FDR	&	1.508	&	0.747	&	0.455	&	&	&	FDR	&	0.506	&	0.326	&	0.225	\\
&	&	CV	&	0.803	&	0.367	&	0.218	&	&	&	CV	&\textbf{0.369}	&\textbf{0.202}	&\textbf{0.127}	\\
&	&	SURE	&	0.836	&	0.383	&	0.225	&	&	&	SURE	&	0.463	&	0.321	&	0.238	\\
&	&	BAMS	&	1.254	&	0.409	&	0.308	&	&	&	BAMS	&	1.055	&	0.261	&	0.177	\\
&	&	LPM	&	4.966	&	1.957	&	0.998	&	&	&	LPM	&	4.966	&	1.957	&	0.998	\\
&	&	Bickel	&	0.689	&\textbf{0.290}	&	0.686	&	&	&	Bickel	&	0.657	&	0.454	&	0.518	\\
&	&	$a=1$	&	0.707	&	0.306	&	0.192	&	&	&	$a=1$	&	0.593	&	0.429	&	0.361	\\
&	&	$a=2$	&	0.680	&	0.322	&	0.255	&	&	&	$a=2$	&	0.617	&	0.449	&	0.409	\\
&	&	$a=5$	&	0.684	&	0.301	&\textbf{0.183}	&	&	&	$a=5$	&	0.733	&	0.585	&	0.629	\\
&	&	$a=10$	&	1.348	&	1.297	&	0.744	&	&	&	$a=10$	&	0.851	&	0.487	&	0.398	\\
&	&	Triang	&\textbf{0.677}	&	0.340	&	0.253	&	&	&	Triang	&	0.624	&	0.425	&	0.425	\\ \hline
																			
&2048	&	Univ	&	1.146	&	0.578	&	0.364	&	&2048	&	Univ	&	0.363	&	0.233	&	0.165	\\
&	&	FDR	&	1.038	&	0.487	&	0.295	&	&	&	FDR	&	0.391	&	0.232	&	0.155	\\
&	&	CV	&	0.551	&	0.252	&	0.146	&	&	&	CV	&\textbf{0.265}	&\textbf{0.141}	&\textbf{0.088}	\\
&	&	SURE	&	0.568	&	0.258	&	0.148	&	&	&	SURE	&	0.365	&	0.236	&	0.168	\\
&	&	BAMS	&	1.085	&	0.275	&	0.184	&	&	&	BAMS	&	0.982	&	0.200	&	0.120	\\
&	&	LPM	&	4.955	&	1.957	&	0.998	&	&	&	LPM	&	4.955	&	1.957	&	0.998	\\
&	&	Bickel	&	0.430	&	0.713	&	0.168	&	&	&	Bickel	&	0.501	&	0.549	&	1.323	\\
&	&	$a=1$	&	0.413	&	0.201	&	0.154	&	&	&	$a=1$	&	0.466	&	0.346	&	0.306	\\
&	&	$a=2$	& \textbf{0.405}	&	0.254	&	0.148	&	&	&	$a=2$	&	0.475	&	0.426	&	0.325	\\
&	&	$a=5$	&	0.423	& \textbf{0.192}	&	0.146	&	&	&	$a=5$	&	0.606	&	0.412	&	0.320	\\
&	&	$a=10$	&	1.681	&	2.523	&	2.509	&	&	&	$a=10$	&	0.613	&	0.385	&	0.315	\\
&	&	Triang	&	0.408	&	0.275	&\textbf{0.145}	&	&	&	Triang	&	0.480	&	0.405	&	0.329	\\ \hline

\end{tabular}
\caption{AMSE of the shrinkage/thresholding rules in the simulation study for Doppler and Heavisine DJ test functions.}\label{tab:sim2}
\end{table}

\begin{figure}[H]
\subfigure[SNR=3.\label{lognormal}]{
\includegraphics[scale=0.5]{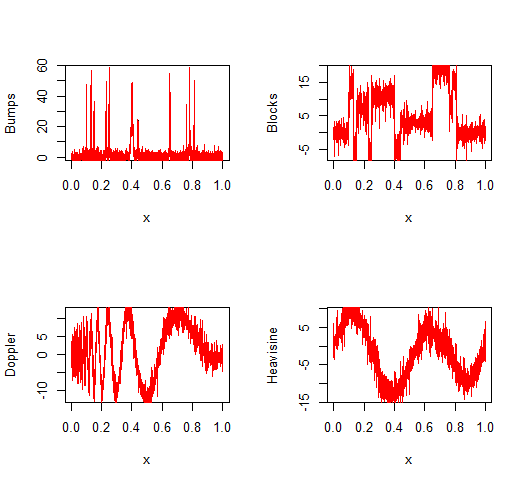}}
\subfigure[SNR=5.\label{blocls}]{
\includegraphics[scale=0.5]{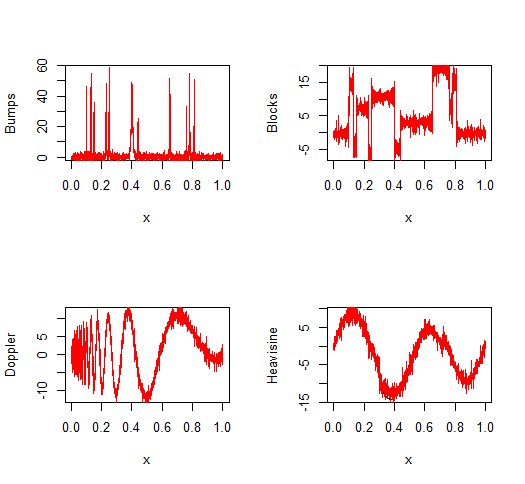}}
\subfigure[SNR=7.\label{blocls}]{
\includegraphics[scale=0.5]{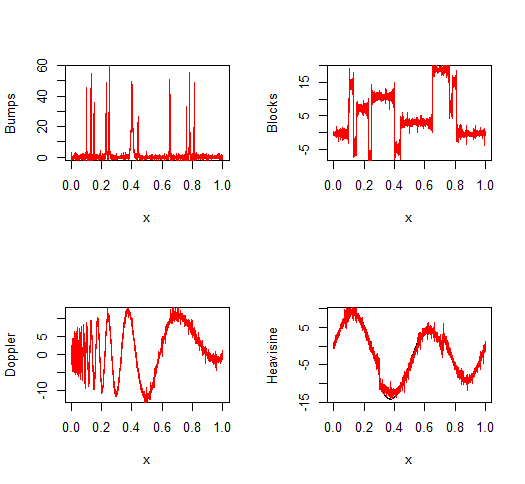}}
\caption{Donoho-Johnstone test functions: Estimates obtained by wavelet shrinkage rule under beta prior with $a=2$ and $n=2048$.}\label{fig:fits}
\end{figure}

\section{Application in Spike-Sorting Data Set}

Spike sorting is a classification procedure of action potentials (spikes) emitted by neurons according to their different forms and amplitudes. Typically, action potentials data for sorting by spike sorting are collected extracellularly by means of electrodes connected at certain locations in the head of animals. It is a method of extreme relevance in Neuroscience due to the possibility of studies about which neurons are present in certain regions of the brain and how they interact.

Once the raw data of action potentials is collected, the first step of the spike sorting procedure is to denoise data to facilitate visualization of spikes and prevent misclassification of a noise as spike. Among several methods used for spike sorting data noise reduction, wavelet based methods are one of the most used. For more details on spike sorting and statistical methods involved in the analysis of characteristic data, one has Pouzat et al. (2002), Lewicki (1998), Shoham et al. (2003), Einevoll et al. (2012) among others. Applications of wavelets in spike sorting occur in the works of Letelier and Webber (2000), Quiroga et al. (2004) and Shalchyan et al. (2012). The purpose here use the beta and triangular shrinkage rules in the wavelet domain for noise reduction.

The original data set, presented in Figure \ref{fig:pot},  has 20,000 neuronal action potentials (spikes) observed over time.  The data set was collected by Kenneth Harris, of \textit{Institute of Neurology, Faculty of Brain Sciences, University College London}, and it is available at https://ifcs.boku.ac.at/repository/data/spike-sorting/index.html.
For the wavelet application, we redice the sample size to  $n=2^{14}=16384$.

\begin{figure}[H]
\centering
\includegraphics[scale=0.60]{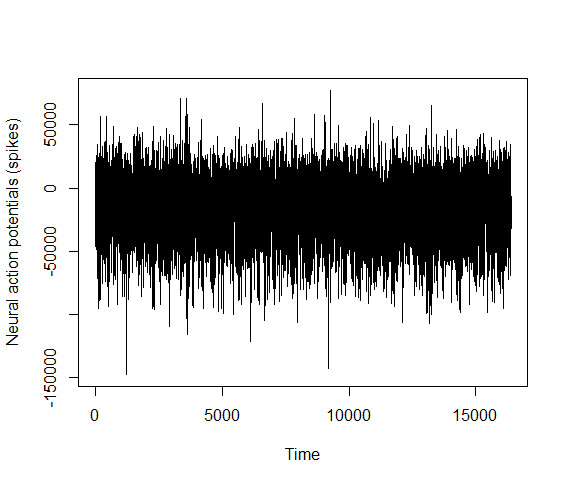}
\caption{Neural action potentials (\textit{spikes}).}\label{fig:pot}
\end{figure}

The shrinkage rule with beta and triangular priors were applied in the wavelet domain. The hyperparameters chosen for $\sigma$, $m$ e $\alpha$ were given according to Section 5, with $\hat{\sigma} = 19913$ and $a=2$ for the beta distribution. Figure \ref{fig:estpot} presents the estimated functions and Figure \ref{fig:potcoef} shows the empirical wavelet coefficients after application of transform with Daubechies filter with ten vanishing moments, $N = 10.$

\begin{figure}[H]
\centering
\subfigure[Estimated action potentials - beta prior shrinkage rule with $a=2$.\label{lognormal}]{
\includegraphics[scale=0.4]{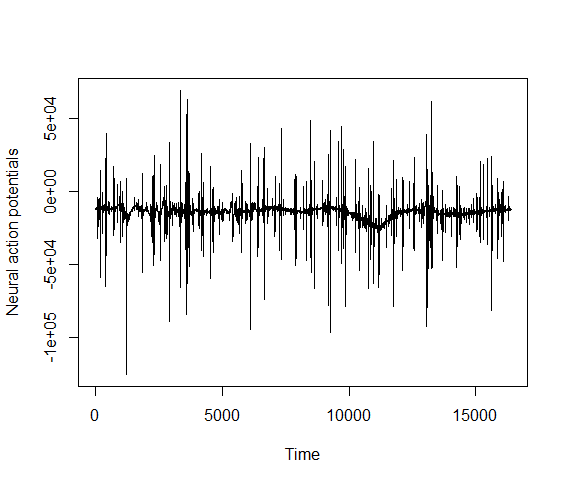}}
\subfigure[Estimated action potentials - triangular prior shrinkage rule.\label{blocls}]{
\includegraphics[scale=0.4]{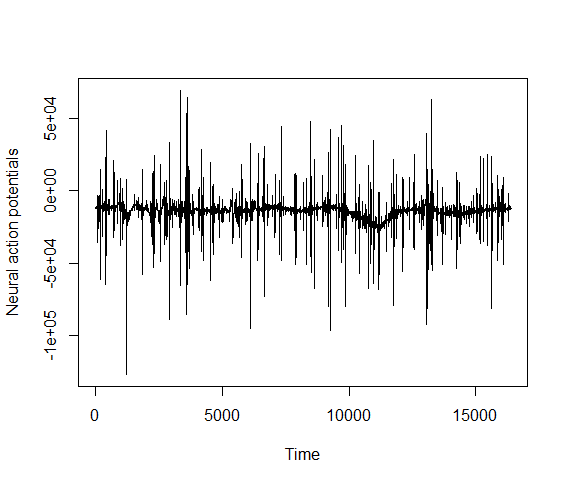}}
\caption{Estimated action potentials - beta prior shrinkage rule with $a=2$ (a) and triangular prior shrinkage rule (b).}\label{fig:estpot}
\end{figure}

\begin{figure}[H]
\centering
\subfigure[Empirical coefficients.\label{lognormal}]{
\includegraphics[scale=0.4]{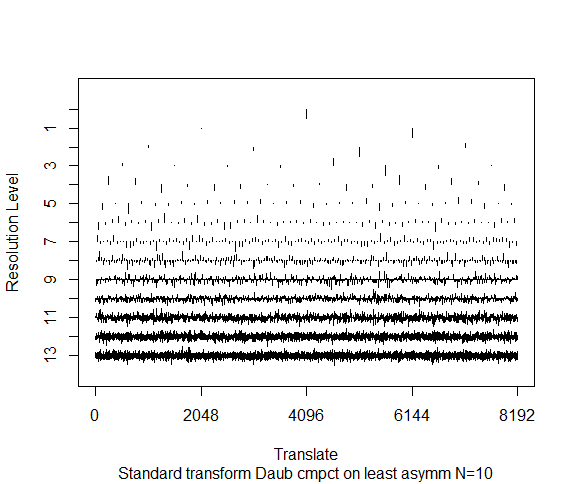}}
\subfigure[Estimated coefficients - beta prior shrinkage rule with $a=2$.\label{blocls}]{
\includegraphics[scale=0.4]{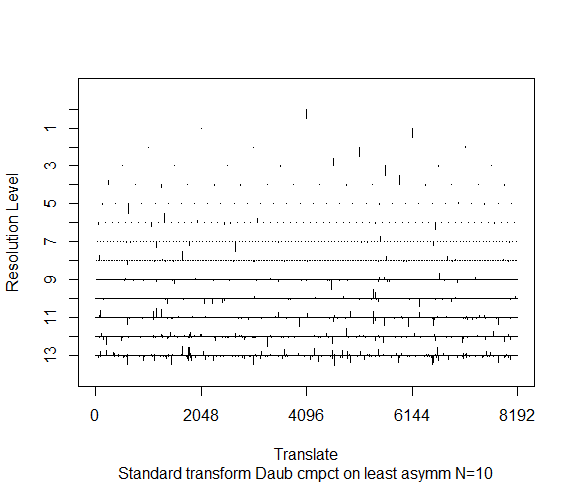}}
\caption{Empirical coefficients (a) and estimated coefficients - beta prior shrinkage rule with $a=2$ (b) of the Spike-Sorting data set.}\label{fig:potcoef}
\end{figure}

\section{Conclusions}

In this paper we proposed the use of a discrete mixture of a point mass at zero and the beta distribution as a prior model to wavelet coefficients. This model allows for the incorporation of prior information about the coefficients for bounded energy signals, which is an advantage in comparison with standard shrinkage/thresholding techniques. Furthermore, the hyperparameters are readily related to the shrinkage level of the associated shrinkage rule, a useful feature which  facilitates their elicitation. 

The results indicate that the shrinkage rules associated with such prior model perform better compared to some standard shrinkage/thresholding techniques for majority of cases and test functions in terms of average mean squared error. The proposed rules are particularly useful for noisy signals with a low signal-to-noise ratio.
Because of this performance and fairly easy computation, we recommend to practitioners the beta and triangular wavelet shrinkage, especially in the cases when the noise dominates the signal.

Further extensions, generalization and new results are planned. The performance of the shrinkage rules on statistical models with distributions with possibly asymmetric support could be considered. The impact of using different wavelet bases in such rules may also be of interest and were not considered here. As an improvement   of the proposed technique, the use of level-dependent hyperparameters in simulation studies and comparisons against the state-of-art techniques, especially for a low SNR may be of interest.

Thus, the proposed model provides a basis on which more sophisticated generalizations can be made. For example, univariate betas can be replaced with their multivariate versions (several such generalizations exist), while hyperparameters can be connected via hyperpriors to address the inter-dependences of empirical wavelet coefficients.

\appendix
\section{Proof of Proposition 4.1}

The proof of Proposition 4.1 requires the following lemma, involving truncated moments of the standard normal distribution.

\begin{lema}\label{lema}
Let $\phi(\cdot)$ and $\Phi(\cdot)$ be the standard normal density and cumulative distribution functions respectively and  $a,b \in \mathbb{R}$. Then,

\begin{enumerate}

\item $\int_{a}^{b}x\phi(x)dx = \phi(a)-\phi(b),$

\item $\int_{a}^{b}x^2\phi(x)dx = [a\phi(a)-b\phi(b)]+[\Phi(b)-\Phi(a)],$

\end{enumerate}
\end{lema}

We now provide the proof of Proposition 4.1.
\begin{proof}
Applying Proposition \ref{prop1} we get
\begin{align*}
\delta_{T}(d) &= \frac{(1-\alpha)\int_{\frac{-m-d}{\sigma}}^{\frac{m-d}{\sigma}}(\sigma u + d)g_{T}(\sigma u + d)\phi(u)du}{\alpha \frac{1}{\sigma}\phi(\frac{d}{\sigma})+(1-\alpha)\int_{\frac{-m-d}{\sigma}}^{\frac{m-d}{\sigma}}g_{T}(\sigma u + d)\phi(u)du} \\
              &= \frac{(1-\alpha)I_2}{\alpha \frac{1}{\sigma}\phi(\frac{d}{\sigma})+(1-\alpha)I_1}.
\end{align*}
For $I_1$, the integral in the denominator, we have that
\begin{align*}
I_1 &= \int_{\frac{-m-d}{\sigma}}^{\frac{m-d}{\sigma}}g_{T}(\sigma u + d)\phi(u)du \\
    &= \int_{\frac{-m-d}{\sigma}}^{\frac{-d}{\sigma}}\frac{[(\sigma u +d)+m]}{m^2}\phi(u)du + \int_{\frac{-d}{\sigma}}^{\frac{m-d}{\sigma}}\frac{[m-(\sigma u +d)]}{m^2}\phi(u)du \\
    &= I_{1}^{*} + I_{1}^{**}.
\end{align*}
We calculate $I_{1}^{*}$ and $I_{1}^{**}$ separately and apply Lemma A.1 to solve the definite integrals, i.e,
\begin{align*}
I_{1}^{*} &= \int_{\frac{-m-d}{\sigma}}^{\frac{-d}{\sigma}}\frac{[(\sigma u +d)+m]}{m^2}\phi(u)du \\
          &= \frac{1}{m^2}\bigg\{\sigma \left[\phi\left(\frac{m+d}{\sigma}\right)-\phi\left(\frac{d}{\sigma}\right)\right] + (d+m)\left[\Phi\left(\frac{-d}{\sigma}\right)-\Phi\left(\frac{-m-d}{\sigma}\right)\right]\bigg\}.
\end{align*}

\begin{align*}
I_{1}^{**} &= \int_{\frac{-d}{\sigma}}^{\frac{m-d}{\sigma}}\frac{[m-(\sigma u +d)]}{m^2}\phi(u)du \\
           &= \frac{1}{m^2}\bigg\{(m-d) \left[\Phi\left(\frac{m-d}{\sigma}\right)-\Phi\left(\frac{-d}{\sigma}\right)\right] - \sigma \left[\phi\left(\frac{d}{\sigma}\right)-\phi\left(\frac{m-d}{\sigma}\right)\right]\bigg\}.
\end{align*}
For $I_2$, the integral in the numerator,
\begin{align*}
I_2 = \sigma \int_{\frac{-m-d}{\sigma}}^{\frac{m-d}{\sigma}}ug_{T}(\sigma u + d)\phi(u)du + d\int_{\frac{-m-d}{\sigma}}^{\frac{m-d}{\sigma}}g_{T}(\sigma u + d)\phi(u)du = \sigma I_{2}^{*} + dI_1.
\end{align*}
Then $I_2$ depends on $I_1$ and $I_2^{*}$. Let is obtain $I_2^{*}$.
\begin{align*}
I_{2}^{*} &= \int_{\frac{-m-d}{\sigma}}^{\frac{-d}{\sigma}}u\frac{[(\sigma u +d)+m]}{m^2}\phi(u)du + \int_{\frac{-d}{\sigma}}^{\frac{m-d}{\sigma}}u\frac{[m-(\sigma u +d)]}{m^2}\phi(u)du \\
          & = \frac{\sigma}{m^2}\left[2\Phi\left(\frac{-d}{\sigma}\right)-\Phi\left(\frac{-m-d}{\sigma}\right)-\Phi\left(\frac{m-d}{\sigma}\right)\right].
\end{align*}
In this way, we obtain $I_2$ as
\begin{align*}
I_2 &= \frac{1}{m^2}\bigg\{d\sigma\left[\phi\left(\frac{m+d}{\sigma}\right)+\phi\left(\frac{md}{\sigma}\right)-2\phi\left(\frac{d}{\sigma}\right)\right]+(\sigma^2+dm+d^2)\Phi\left(\frac{m+d}{\sigma}\right)+ \\
    & \quad +(\sigma^2-dm+d^2)\Phi\left(\frac{d-m}{\sigma}\right)-2(\sigma^2+d^2)\Phi\left(\frac{d}{\sigma}\right)\bigg\}.
\end{align*}
Finally, substituting the integrals $I_1$ and $I_2$ into the original expression of $\delta_T(d)$ and for convenient definitions of $S_1(d)$ and $S_2(d)$, we have
\begin{equation}
\delta_T(d) = \frac{(1-\alpha)I_2}{\alpha \frac{1}{\sigma}\phi(\frac{d}{\sigma})+(1-\alpha)I_1}=\frac{(1-\alpha)S_1(d)}{\frac{\alpha m^2}{\sigma}\phi(\frac{d}{\sigma}) + (1-\alpha)S_2(d)}. \nonumber
\end{equation}

\end{proof}

\end{document}